\documentclass[a4paper,UKenglish,cleveref, autoref, thm-restate]{lipics-v2019}

\usepackage[utf8]{inputenc}
\usepackage{amsmath, complexity}
\usepackage{amsfonts, graphicx}
\usepackage{amssymb}
\usepackage{float}
\usepackage{graphicx}
\usepackage{thmtools}
\newtheorem{observation}{Observation}

\nolinenumbers
\newcommand{\ie}{i.e., }
\newcommand{\etal}{et al.}

\newcommand{\perm}{\ensuremath{{\sf perm}}}
\newcommand{\rof}{\ensuremath{{\sf ROF}}}
\newcommand{\roabp}{\ensuremath{{\sf ROABP}}}
\newcommand{\PIT}{\textsf{PIT }}

\newcommand{\var}{\ensuremath{{\sf var}}}
\newcommand{\rank}{\ensuremath{{\sf rank}_{\varphi}}}
\newcommand{\brank}{\ensuremath{{\sf rank}_{\hat{\varphi}}}}

\newcommand{\ry}{{\sf RY}}
\newcommand{\distr}{\ensuremath{{\cal D}_{\cal B}}}

\title{Limitations of Sums of Bounded Read Formulas}
\author{Purnata Ghosal}{Department of Computer Science and Engineering, IIT Madras, Chennai, India}{purnata@cse.iitm.ac.in}{}{}

\author{B. V. Raghavendra Rao}{Department of Computer Science and Engineering, IIT Madras, Chennai, India}{bvrr@cse.iitm.ac.in}{}{}

\authorrunning{P. Ghosal and B.\,V. Raghavendra Rao} 

\Copyright{Purnata Ghosal and B. V. Raghavendra Rao} 



\EventEditors{John Q. Open and Joan R. Access}
\EventNoEds{2}
\EventLongTitle{42nd Conference on Very Important Topics (CVIT 2016)}
\EventShortTitle{CVIT 2016}
\EventAcronym{CVIT}
\EventYear{2016}
\EventDate{December 24--27, 2016}
\EventLocation{Little Whinging, United Kingdom}
\EventLogo{}
\SeriesVolume{42}
\ArticleNo{23}

\begin{document}

\maketitle

\begin{abstract}
Proving super polynomial size lower bounds for  various classes of arithmetic circuits computing explicit polynomials is a very important and challenging task in algebraic complexity theory.    We study representation of polynomials as  sums of weaker models such as read once formulas (ROFs) and read once oblivious algebraic branching programs (ROABPs). We prove:
\begin{enumerate}
\item[(1)] An exponential separation between sum of ROFs and read-$k$ formulas for  some constant $k$.
\item[(2)] A sub-exponential separation between sum of ROABPs and syntactic multilinear ABPs.
\end{enumerate}

Our results are based on analysis of the partial derivative matrix under different distributions. These results highlight richness of bounded read restrictions in arithmetic formulas and ABPs. 

Finally, we consider a generalization of multilinear ROABPs known as strict-interval ABPs defined in [Ramya-Rao, MFCS2019]. We show that strict-interval ABPs are equivalent to ROABPs upto a polynomial size blow up. In contrast, we show that interval formulas are different from ROFs and also admit depth reduction which is not known in the case of strict-interval ABPs.

\end{abstract}
\section{Introduction}
\label{sec:intro}
Polynomials are one of the fundamental mathematical objects and have wide applications in Computer Science. Algebraic Complexity Theory aims at a classification of polynomials based on their computational complexity. In his seminal work, Valiant~\cite{Val79} laid foundations of Algebraic Complexity Theory and popularized arithmetic circuits as a natural model of computation for polynomials. He proposed the permanent polynomial  $\perm_n$:

$$\perm_n = \sum_{\pi \in S_n} \prod_{i=1}^n x_{i\pi(i)}, $$ 

as the primary  representative of intractability in algebraic computation. In fact, Valiant\cite{Val79} conjectured that the complexity of computing $\perm_n$ by arithmetic circuits is different from that of the determinant function which is now known as Valiant's hypothesis.  

One of the important offshoots of Valiant's hypothesis is the arithmetic circuit lower bound problem: prove a super polynomial lower bound on the size of an arithmetic circuit computing an explicit polynomial of polynomial degree.  Here,  an  explicit polynomial is one whose coefficients are efficiently computable.  Baur and Strassen~\cite{BS83} obtained a super linear lower bound on the size of any arithmetic circuit computing the sum of powers of variables. This  is the best known size lower bound for general classes of arithmetic circuits. 

Lack of improvements in the size lower bounds for general arithmetic circuits lead the community to investigate restrictions on arithmetic circuits. Restrictions considered in the literature  can be broadly classified into two categories: syntactic and semantic. Syntactictic restrictions considered in the literature  include     restriction on fan-out \ie arithmetic formulas, restriction on depth  \ie bounded depth circuits~\cite{GK98,GR00,SW01}, and  the related model of   algebraic branching programs.   Semantic restrictions  include  monotone arithmetic circuits~\cite{JS82,Yeh19,Sri19}, homogeneous circuits~\cite{Bur}, multilinear circuits~\cite{Raz06} and  noncommutative computation~\cite{Nis91}. 

Grigoriev and Razborov~\cite{GR00} obtained an exponential lower bound for the size of  a depth three arithmetic circuit computing the determinant or and permanent over finite fields.  In contrast, only almost cubic lower bound is known over infinite fields~\cite{KST16}. Explaining the lack of progress on proving lower bounds even in the case of depth four circuits, Agrawal and Vinay~\cite{AV08} showed that an exponential lower bound for the size of depth four circuits implies Valiant's hypothesis over any field.   This lead to intense research efforts in proving lower bounds for the size of constant depth circuits, the reader is referred to an excellent survey by Saptharishi \etal ~\cite{RP16} for details.  

Recall that a polynomial $p\in \mathbb{F}[x_1,\ldots,x_n]$ is said to be {\em multilinear} if every monomial in $p$  with non-zero coefficient is square-free. An arithmetic circuit is said to be multilinear if every gate in the circuit computes a multilinear polynomial. Multilinear circuits are natural models for computing multilinear polynomials. Raz~\cite{Raz09} obtained super polynomial lower bounds on the size of multilinear formulas computing the determinant or permanent.  Further, he gave a super polynomial separation between multilinear formulas and circuits~\cite{Raz06}. In fact, Raz~\cite{Raz09}  considered a syntactic  version of multilinear circuits known as {\em syntactic multilinear } circuits.  An arithmetic  circuit $C$ is said to be syntactic multilinear, if for every product gate $g=g_1 \times g_2$  the sub-circuits rooted at $g_1$ and $g_2$ are variable disjoint.   The syntactic version has an advantage that the restriction can be verified by examining the circuit whereas there is no efficient algorithm for testing if a circuit is  multilinear or not.  Following Raz's work, there has been significant interest in proving lower bounds on the size of syntactic multilinear circuits. Exponential separation of constant depth multilinear circuits is known~\cite{CELS18}, while the best known lower bound for unbounded depth syntactic multilinear circuits is only almost quadratic~\cite{AKV20}.

An {\em Algebraic Branching Program} (ABP)  is a model of computation for polynomials that generalize arithmetic formulas and were studied by Ben-Or and Cleve~\cite{BC92} who showed that ABPs of constant width are equivalent to arithmetic formulas. Nisan~\cite{Nis91}  proved exponential size lower bound for the size of an ABP computing the permanent  when the variables are non-commutative.   It is known that polynomial families computed by  ABPs are   the same as families of polynomials computed by skew circuits, a restriction of arithmetic circuits where every product gate can have at most one non-input gate as a predecessor~\cite{MP08}. Further, skew arithmetic circuits are known to characterize the complexity of determinant~\cite{Tod92}. Despite their simplicity compared to arithmetic circuits, the best known lower bound for size of ABPs is only quadratic~\cite{Kum19, CKSV20}. Even with the restriction of syntactic multilinearity, the best known size lower bound for ABPs is only quadratic~\cite{Jan08}. However, a super polynomial separation between syntactically multilinear formulas and ABPs is known~\cite{DMPY12}.   
 
Proving super quadratic size lower bounds for syntactic multilinear ABPs (smABPs for short) remains a challenging task.  Given that there is no promising approach yet to prove super quadratic size lower bounds for smABPs, it is imperative to consider further structural restrictions on smABPs and formulas to develop finer insights into the difficulty of the problem. Following the works in~\cite{RR18,RR19a,RR19}, we study syntactic multilinear formulas and smABPs with restrictions on the number of reads of variables and the order in which variables appear in a smABP. 

\subparagraph*{Models and Results:}  
\noindent\textbf{(1) Sum of ROFs:} 
A {\em read-once formula} (ROF) is a formula where every variable occurs exactly once as a leaf label. ROFs are syntactic multilinear by definition and have received wide attention in the literature. Volkovich~\cite{Vol16} gave a complete characterization of polynomials computed by ROFs. Further, Minahan and Volkovich~\cite{MV18} obtained a complete derandomization of the polynomial identity testing problem on ROFs. While most of the multilinear polynomials are not computable by ROFs~\cite{Vol16}, sum of ROFs, denoted by $\Sigma\cdot\rof$ is a natural model of computation for multilinear polynomials. Shpilka and Volkovich showed that a restricted form of $\Sigma\cdot\rof$ requires linear summands to compute the monomial $x_1x_2\cdots x_n$. Further, Mahajan and Tawari~\cite{MT18} obtained a tight lower bound on the size of $\Sigma\cdot\rof$ computing an elementary symmetric polynomial.  Ramya and Rao~\cite{RR19a}  obtained an exponential lower bound on the number of ROFs required to compute a polynomial in $\VP$.  In this article, we  improve the lower bound in~\cite{RR19a} to obtain an exponential separation between read-$k$ formulas and $\Sigma\cdot\rof$ for  a sufficiently large constant $k$. Formally, we prove:

\begin{restatable}{theorem}{sumrof}
There is constant $k>0$ and  a family of multilinear polynomials $f_{\textsf{PRY}}$ computable by  read-$k$ formulas such that if  $ f_{\textsf{PRY}} = f_1+f_2+\dots + f_s$, where $f_1,\ldots, f_s$ are ROFs,  then $s= 2^{\Omega(n)}$.
\end{restatable}  

\noindent\textbf{(2) Sum of ROABPs: } 
A natural generalization of ROFs are read-once oblivious branching programs (ROABPs). In an ROABP, a layer reads at most one variable and every variable occurs in exactly one layer. Arguments in~\cite{Nis91} imply that any ROABP computing the permanent and determinant requires exponential size. Kayal \etal~\cite{KNS16} obtain an exponential separation between the size of ROABPs and depth three multilinear formulas. 
In~\cite{RR18}, an exponential lower bound for the sum of ROABPs computing a polynomial in $\VP$ is given. We improve this bound to obtain a super polynomial separation between sum of ROABPs and smABPs:
\begin{restatable}{theorem}{sumroabp}
\label{thm:sumroabp-lb}
There is a multilinear polynomial family $\hat f$ computable by smABPs of polynomial size such that  if $\hat{f}=f_1+\ldots+f_s$, each $f_i \in \mathbb{F}[X]$ being computable by a ROABP of size $\text{poly}(n)$, then $s = \text{exp}(\Omega(n^\epsilon))$ for  some $\epsilon<1/500$.
\end{restatable}  

\noindent\textbf{(3) Strict-interval ABPs and Interval formulas:} 
It may be noted that any sub-program of a ROABP computes a polynomial in an interval $\{x_i,\ldots, x_j\}$ of variables for some $i<j$. A natural generalization of ROABPs would be to consider smABPs where every sub-program computes a polynomial in some interval of variables, while a variable can occur in multiple layers. These are known as {\em interval} ABPs and were studied by Arvind and Raja~\cite{AR16} who obtained a conditional lower bound for the size of  interval ABPs. Ramya and Rao~\cite{RR19} obtained an exponential lower bound for a special case of interval ABPs known as {\em strict-interval } ABPs.  We show that strict-interval ABPs are equivalent to ROABPs upto polynomial size:

\begin{restatable}{theorem}{intervalabp}
The class of strict-interval ABPs is equivalent to the class of ROABPs.
\end{restatable}

Finally, we examine the restriction of intervals in syntactic multilinear formulas. We show that unlike ROFs, interval formulas can be depth reduced (Theorem~\ref{thm:intf-depthred}). 

\subparagraph*{Related Work:} 
To the best of our knowledge, Theorem~\ref{thm:sumrof} is the first exponential separation between bounded read formulas and $\Sigma\cdot\rof$. Prior to this, only a linear separation between bounded read formulas and $\Sigma\cdot\rof$ was known~\cite{AMV11}. 

Ramya and Rao~\cite{RR19a} obtain an exponential separation between $\Sigma\cdot\rof$ and multilinear $\VP$. Our result is an extension of this result for the case of a simpler polynomial computable by bounded read formulas.  Mahajan and Tawari~\cite{MT18} obtain tight linear lower bound for $\Sigma\cdot\rof$ computing an elementary symmetric polynomial.  

Kayal, Nair and Saha~\cite{KNS16} obtain a separation between ROABPs and multilinear depth three circuits. The authors define a polynomial, efficiently computed by set multilinear depth three circuits, that has an exponential size ROABP computing it. This polynomial can be expressed as a sum of three ROFs. Later, Ramya and Rao~\cite{RR18} obtain a sub-exponential lower bound against the model of $\Sigma \cdot \roabp$ computing the polynomial defined by Raz and Yehudayoff~\cite{RY08}. Dvir \etal ~\cite{DMPY12} obtain a super-polynomial lower bound on the size of syntactic multilinear formulas computing a polynomial that can be efficiently computed by smABPs. We use the polynomial defined by \cite{DMPY12} and adapt their techniques to obtain a separation between smABPs and $\Sigma \cdot \roabp$. 
 
 

\subparagraph*{Organization of the Paper:} Section~\ref{sec:prelims} contains basic definitions of the models of computations, concepts and explicit polynomials used in the rest of the paper. The rest of the sections each describe results with respect to a particular bounded-read model. Section~\ref{sec:sumrof} describes the lower bound on the $\Sigma\cdot \rof$ model and Section~\ref{sec:sumroabp-lb} describes the lower bound on the $\Sigma\cdot \text{ROABP}$ model which follows using the same arguments as in the work of Dvir \etal \cite{DMPY12}. Section~\ref{sec:sintABP} shows that strict-interval ABP is a fresh way to look at ROABPs since the two models are equivalent. In Section~\ref{sec:intformula} we see that Brent's depth reduction result (\cite{Bre74}) holds for the class of interval formulas.

\section{Preliminaries}
\label{sec:prelims}
In this section, we present necessary definitions and notations. For more details, reader is referred to excellent surveys by Shpilka and Yehudayoff~\cite{SY10} and by Saptharishi \etal ~\cite{RP16}. 

\subparagraph*{Arithmetic Circuits: }
Let $X=\{x_1,\ldots, x_n\}$ be a set of variables. 
An arithmetic circuit $C$ over a field $\mathbb{F}$ with input $X$ is a directed acyclic graph (DAG) where the nodes have in-degree zero or two. The nodes of in-degree zero are called input gates  and  are labeled by  elements from $X \cup \mathbb{F}$.  Non-input gates of $C$ are called internal gates and are labeled from $\{+, times\}$.  Nodes of out degree zero are called output gates. Typically, a circuit has a single output gate.    Every gate $v$  in $C$ naturally computes a polynomial $f_v \in \mathbb{F}[X]$. The   polynomial computed by $C$ is the polynomial represented at  its output gate.  
The {\em size} of a circuit denoted by $\text{size}(C)$, is the number of gates in it,  and {\em depth} is  the length of the longest root to leaf path in $C$, denoted by $\text{depth}(C)$. An  {\em arithmetic formula} is a  circuit where the underlying undirected graph is a tree. For a gate $v$ in $C$, let $\var(v)$ denote the set of all variables that appear as leaf labels in the sub-circuit rooted at $v$.

Multilinear polynomials are polynomials such that in every monomial, the degree of a variable is either $0$ or $1$. Multilinear circuits, where every gate in the circuit computes a multilinear polynomial, are  natural models of computation for multilinear polynomials. A circuit $C$ is said to be {\em syntactic multilinear} if for every product gate $v= v_1\times v_2$ in $C$, we have $\var(v_1) \cap \var(v_2) =\emptyset$. By definition, a syntactic multilinear circuit is also multilinear and computes  a multilinear polynomial. 

An arithmetic formula $F$ is said to be a  {\em read-once formula} (ROF in short) if every input variable in $X$ labels at most one input gate in $F$.

 {\em Algebraic branching program} (ABP in short) is a model of computation of polynomials defined as analogous to the branching program model of computation for Boolean functions. An ABP $P$ is a layered DAG with layers $L_0,\ldots,L_m$ such  $L_0 = \{s\}$ and $L_m = \{t\}$ where $s$ is the start node and $t$ is the terminal node.
 Each edge is labeled by an element in $X \cup \mathbb{F}$.    The output of the ABP $P$ is the polynomial $p= \sum_{\rho \text{is a $s$ to $t$ path}} {\sf wt}(\rho)$, where ${\sf wt}(\rho)$ is  the product of edge labels in the path $\rho$. Further, for any two nodes $u$ and $v$ let $[u,v]_P$ denote the polynomial computed by the subprogram  $P_{uv}$ of $P$ with $u$ as the start node and $v$ as the terminal node. Let $X_{uv}$ denote the set of variables that appear as edge labels in the subprogram  $P_{uv}$. The size of an ABP $P$, denoted by $\text{size}(P)$ is the number of nodes in $P$.

In a {\em syntactic multilinear} ABP (smABP),  every $s$ to $t$ path   reads any input variables at most once. An ABP is  {\em oblivious} if every layer reads at most one variable.   A {\em read-once oblivious ABP} (ROABP) is an oblivious syntactic multilinear ABP where every variable can appear in at most one variable \ie for every $i$, there is at most one layer $j_i$ such that $x_i$ occurs as a label on the edges from $L_{j_i}$ to $L_{j_i +1 }$. 

An interval on the set $\{1, \ldots, n\}$ with end-points $i,j \in [n]$, can be defined as $I= [i,j], ~i<j$, where $I=\{\ell \mid i,j \in [n], i \le \ell \le j\}$. An {\em interval of variables} $X_{ij}$ is defined such that $X_{ij} \subseteq  \{x_\ell \mid \ell \in I, ~I=[i,j]\}$, where $I$ is an interval on the set $\{1, \dots, n\}$. For an ordering $\pi \in S_n$, we define a $\pi$-interval of variables, $X_{ij} \subseteq \{x_{\pi(i)}, x_{\pi(i+1)}, \ldots, x_{\pi(j)}\}$. 
In~\cite{AR16}, Arvind and Raja defined a sub-class of syntactically multilinear ABPs known as {\em interval ABPs} and proved lower bounds against the same. Later, \cite{RR19} defined a further restricted version of interval ABPs, denoted by strict-interval ABPs, defined as follows.

\begin{definition}(\cite{RR19})\label{def:sint-ABP}
A strict interval ABP $P$ is a syntactically multilinear ABP where we have the following:
\begin{enumerate}
\item 
For any pair of nodes $u$ and $v$ in $P$, the indices of variables occurring in the sub-program $[u,v]_P$ is contained in some  $\pi$-interval $I_{uv}$  called the {\em associated interval} of $[u,v]_P$; and  

\item for any pairs of sub-programs of the form $[u,v]_P, ~[v,w ]_P$, the  associated$\pi$-intervals of variables are disjoint, 
\ie $I_{uv} \cap I_{vw} = \emptyset$.
\end{enumerate}
\end{definition}

It may be noted that in a strict interval ABP,   intervals associated with each sub-program need not be unique. We assume that the intervals associated are largest intervals with respect to set inclusion such that condition 2 in the definition above is satisfied.  

\subparagraph*{The Partial Derivative Matrix:}

We need the notion of partial derivative matrices introduced by Raz~\cite{Raz09} and Nisan~\cite{Nis91} as primary measure of complexity for multilinear polynomials. The partial derivative matrix of a polynomial $f\in \mathbb{X}$ defined based on a partition $\varphi:X \to Y \cup Z$ of the $X$ into two parts. We follow the definition in~\cite{Raz09}: 

\begin{definition}(Raz~\cite{Raz09}) Let $\varphi:X \to Y \cup Z$ be a partition of the input variables in two parts. Let $\mathcal{M}_Y$, $\mathcal{M}_Z$ be the sets of all possible multilinear  monomials in the  variables in $Y$ and $Z$ respectively. Then we construct the partial derivative matrix $M_{f^\varphi}$ for a multilinear polynomial $f$ under the partition $\varphi$ such that the rows of the matrix are indexed by monomials $m_i \in \mathcal{M}_Y$, the columns by monomials $s_j \in \mathcal{M}_Z$ and entry $M_{f^\varphi}(i,j)=c_{ij}$, $c_{ij}$ being the coefficient of the monomial $m_i\cdot s_j$ in $f$. We denote by $\rank(f)$ the rank of the matrix $M_{f^\varphi}$.
\end{definition}
We call $\varphi$   an equi-partition when $|X|=n$, $n$ even and $|Y|=|Z|=n/2$.

Raz~\cite{Raz09} showed the following fundamental property of $\rank$: 

\begin{lemma}\label{lem:rank-prop}
Let $g$ and $h$ be multilinear polynomials in $\mathbb{F}[X]$. Then, $\forall \varphi:X \to Y \cup Z$, we have the following.
\begin{description}
\item[Sub-additivity:]  $\rank(g+h) \le \rank(g) + \rank(h)$, and
\item[Sub-multiplicativity:]   $\rank(gh) \le \rank(g)\times \rank(h)$,
\end{description}
In both the cases, equality holds when  $\var(g) \cap \var(h) = \emptyset$. 
\end{lemma}


\subparagraph*{Two Explicit Polynomials:}
Polynomials that exhibit maximum  rank ofthe partial derivative matrix under all or a large fraction of  equi-partitions   can be thought of as {\em high complexity} or {\em hard} polynomials.  We need two such families found in the literature. 

Raz and Yehudayoff~\cite{RY08} defined a multilinear  in $\VP$. To define this polynomial we denote an interval $\{a\mid i\le a\le j, a \in \mathbb{N}\}, ~i,j\in \mathbb{N}$ by $[i,j]$, and consider the sets of variables $X=\{x_1, \ldots, x_{2n}\}$, $W=\{w_{i,\ell,j}\}_{i,\ell,j \in [2n]}$. We denote it as the Raz-Yehudayoff polynomial and define it as follows.

\begin{definition}[{\bf Raz-Yehudayoff polynomial}, \cite{RY08}]\label{def:RYpoly}
Let us consider $f_{ij} \in \mathbb{F}[X,W]$ defined over the interval $[i,j]$. 
For $i\le j$, the polynomial $f_{ij}$ is defined inductively as follows. 
If $j-i=0$, then $f_{ij}=0$. For $|j-i|>0$,  
$$f_{ij} = (1+x_ix_j)f_{i+1,j-1} + \sum_{\ell \in [i+1, j-2]} w_{i,\ell,j} f_{i,\ell}f_{\ell+1,j},$$
where we assume without loss of generality, lengths of $[i,\ell], ~[\ell+1,j]$ are even and smaller than $[i,j]$. We define $f_{1,2n}$ as the Raz-Yehudayoff polynomial $f_{{\sf RY}}$.
\end{definition}
Note that, $f_{\sf RY}$ can be defined over any subset $X'\subseteq X$ such that $|X'|$ is even, by considering the induced ordering of variables in $X'$ and considering intervals accordingly.  We denote this polynomial as $f_{\sf RY}(X')$ for $X' \subseteq X$.
Raz and Yehudayoff showed: 
\begin{proposition}(\cite{RY08})\label{prop:RYrank}
Let $\mathbb{G}=\mathbb{F}(W)$ be the field of rational functions over the field $\mathbb{F}$ and the set of variables $W$. Then  for every equi-partition $\varphi: X \to Y\cup Z$,  $\rank(f_{{\sf RY}}) = 2^{n/2}$.
\end{proposition}

Dvir \etal~\cite{DMPY12} defined a polynomial that is hard \ie full rank with respect to a special class of partitions called {\em arc-partitions}. Suppose $X=\{x_0,\ldots, x_{n-1}\}$ be identified with the set $V=\{0, \ldots, n-1\}$. For $i,j \in V$,  the set $[i,j]=\{i, (i+1)\mod n, (i+2) \mod n,  \ldots, j\}$  is called the {\em arc} from  $i$ to $j$.  
An arc pairing is a distribution on the set of all  pairings (\ie perfect matchings)  on $V$ obtained  in $n/2$ steps as follows. 
 Assuming a pairing $(P_1,\ldots,P_t)$ constructed in $t<n/2$ steps, where $P_1=(0,1)$, $[L_t,R_t]$ is the interval spanned by $\cup_{i\in [t]} P_i$ and the random pair $P_{t+1}$ is constructed such that 
\[P_{t+1} =\begin{cases} &(L_t-2, L_t-1) \text{ with probability }1/3,\\
&(L_t-1, R_t+1) \text{ with probability }1/3,\\
&(R_t+1, R_t+2) \text{ with probability }1/3,\\
\end{cases} \]  and therefore, $[L_{t+1},R_{t+1}]=[L_t,R_t] \cup P_{t+1}$. 

Given a pairing ${\cal P} = \{P_1,\ldots, P_{n/2}\}$ of $V$, there are exactly $2^{n/2}$ partitions of $X$, by assigning $\varphi(x_i) \in Y $ and $\varphi(x_{j}) \in Z$ or  $\varphi(x_i) \in Z $ and $\varphi(x_{j}) \in Y$ independently for each pair $(i,j) \in { \cal P}$. An {\em arc partition} is a distribution on all partitions obtained by sampling an arc pairing as defined above and sampling a partition corresponding to the pairing uniformly at random. We denote this distribution on partitions by ${\mathcal{D}}$.  For a pairing ${\cal P} = \{P_1,\ldots, P_{n/2}\}$ let $M_{{\cal P}}$ be the degree $n/2$ polynomial $\prod_{i=1}^{n/2}(x_{\ell_i}+x_{r_i})$ where $P_i =(\ell_i, r_i)$.  Dvir \etal~\cite{DMPY12} defined the arc full rank polynomial $\widehat{f} = \sum_{{\cal P} \in {\cal D}} \lambda_{{\cal P}}M_{{\cal P}}$, where $\lambda_{\cal P}$ is a formal variable. Dvir \etal~\cite{DMPY12} showed: 

\begin{proposition}\cite{DMPY12}
\label{prop:DMPYrank}
The polynomial $\widehat{f}$ can be computed by a polynomial size smABP and  for every $\varphi \in {\cal D}$,
$\rank(\widehat{f}) = 2^{n/2}$ over a suitable field extension $\mathbb{G}$ of $\mathbb{F}$.
\end{proposition}
%
%
%

Now that we are familiar with most of the definitions required for an understanding of the results in this paper, we proceed to discuss our results. 

\section{Sum of ROFs}
\label{sec:sumrof}
In~\cite{RR19a}, Ramya and Rao show an exponential lower bound for the sum of ROFs computing a polynomial in $\VP$. While this establishes a super polynomial separation between $\Sigma\cdot \rof$ and syntactic multilinear formulas, it is interesting to see if this separation is exponential.  In this section we obtain such an exponential separation. In fact, we show that there is an exponential separation between syntactic multilinear read-$k$ formula and $\Sigma\cdot\rof$.  We begin with the construction of a hard polynomial computable be a read-$k$ formula for a large enough constant $k$.

\subparagraph*{A full rank Polynomial:}
Let $X = \{x_1,\ldots, x_n\}$ be the set of input variables of the hard polynomial such that $4\mid n$. Let $f_\ry(X')$ to denote the Raz-Yehudayoff polynomial defined on the variable set $X'$ of even size, where $X'$ is an arbitrary subset of $X$.

Let $r=\Theta(1)$ be a sufficiently large integer factor of $n$ such that  $r$ and  $n/r$ are both  even. For $1\le i\le n/r$, let $B_i=\{x_{(i-1)r+1 },  \ldots, x_{ir}\}$ and ${\cal B}$ denote the partition $B_{1} \cup B_2 \cup \dots \cup B_{n/r}$ of $X$. The  polynomial $f_{\textsf{PRY}}$ is defined as follows:
\begin{align}\label{eq:prod-RYpoly}
f_{\textsf{PRY}} &{=} f_\ry(B_1)\cdot f_\ry(B_2)\cdots f_\ry(B_{n/r}).
\end{align}  
By definition of the polynomial $f_{\textsf{PRY}}$, it can be computed by a constant-width ROABP of polynomial size as well as by a read-$k$ formula where $k=2^{O(r)}$.

In order to prove a lower bound against a class of circuits computing the polynomial $f_{\textsf{PRY}}$, we consider the complexity measure of the rank of partial derivative matrix. 
Like in \cite{Raz06} and many follow-up results, we analyse the rank of the partial derivative matrix of $f_{\sf PRY}$ under a random partition. The reader might have already noticed that there are equi-partitions under which the $\rank( f_{\sf PRY}) = 1$. Thus, we need a different distribution on the equi-partitions under which $f_{\sf PRY}$ has full rank with probability $1$. In fact, under any partition $\varphi$, which  induces an equi-partition on each of the variable blocks $B_i$, we have $\rank(f_{\sf PRY}) = 2^{n/2}$, \ie full rank. We define $\mathcal{D}_B$ as the uniform distribution on all such partitions. Formally, we have:

\begin{definition}(Distribution ${\cal D}_{\cal B}$) \label{def:distr-Db}
The distribution ${\cal D}_{\cal B}$ is the distribution on the set of all equi-partitions $\hat{\varphi}$ of $X$ obtained by independently sampling an equi-partition $\varphi_i$ of each variable blocks $B_i$, for all $i$ such that $1\le i \le n/r$. We express $\hat{\varphi}$ as $\hat{\varphi}=\varphi_1 \circ \ldots \circ \varphi_{n/r}$.
\end{definition}

For any partition $\varphi$ in the support of $\distr$, we argue that the polynomial $f_{\sf PRY}$ has full rank:
\begin{observation}
\label{obs:full-rank}
For any $\varphi\sim \distr$, $\rank(f_{\textsf{PRY}}) = 2^{n/2}$ with probability $1$.
\end{observation} 
\begin{proof}
Let us fix an equi-partition function $\hat{\varphi} \sim \distr$, $\hat{\varphi}: X \to Y \cup Z$. Let $t=r$. Considering $f_\ry(X')$ where $|X'|=t$ and $t$ is even, we can prove the partial derivative matrix of $f_\ry(X')$ has rank $2^{n/2}$ under $\hat{\varphi}$ by induction on $t$. By definition of $f_\ry$, for $t=2$ we have $f_\ry=0$. 

So, for the higher values of $t$, we see the term $(1+x_1x_t)$ and $f_{2,t-1}$ are variable disjoint, where $(1+x_1x_t)$ has rank $\le 2$, and by the induction hypothesis, $f_{2,t-1}$ has rank $2^{t/2-1}$. Also, by induction hypothesis, for any $\ell$,   the ranks of partial derivative matrices  of $f_{1,\ell}$ and $f_{\ell+1,t}$ are $2^{\ell/2}$ and $2^{(t-\ell)/2}$ respectively. 

When $\hat{\varphi}(x_1)\in Y$ and $\hat{\varphi}(x_t) \in Z$, we set $w_{1,\ell,t}=0$ for all $\ell \in [2,t-1]$ and $\brank(f_{1,t}) = \brank(1+x_1x_t)\cdot \brank{f_{2,t-1}}= 2\cdot 2^{(t/2-1)}=2^{t/2}$. When $\hat{\varphi}(x_1)\in Y$ and $\hat{\varphi}(x_t) \in Y$, for an arbitrary $\ell \in [t]$ we set $w_{1,\ell,t}=1$ and we have $\brank(f_{1,t})=\brank(f_{1,\ell})\cdot \brank(f_{\ell+1,t})= 2^{t/2}$, since $\hat{\varphi}$ is an equi-partition. 

By sub-additivity of rank, and since $B_i, ~i\in [n/r]$ are disjoint sets of variables, we have $\brank(f_{\textsf{PRY}})=\prod_{i\in [n/r]} \brank(f_\ry(B_i)) = \prod_{i\in [n/r]} 2^{t/2} = 2^{tn/2r} = 2^{n/2}$.
\end{proof}

\subsection{Rank Upper Bound on ROFs}
In the following, we argue that the polynomial $h$ cannot be computed by sum of ROFs of sub-exponential size. More formally,
\begin{theorem}
\label{thm:sumrof}
Let $f_1,\ldots, f_s$ be read-once polynomials such that $f_{\textsf{PRY}} = f_1+f_2+\dots + f_s$, then $s= 2^{\Omega(n)}$.
\end{theorem}

We use the method of obtaining an upper bound on the rank of partial derivative matrix for ROFs with respect to a random partition developed by \cite{RR19a}. Though the argument in \cite{RR19a} works for an equi-partition sampled uniformly at random, we show their structural analysis of ROFs can be extended to the case of our distribution $\distr$.  We begin with the notations used in \cite{RR19a} for the categorisation of the gates in a read-once formula $F$. (In this categorisation, the authors have only considered gates with at least one input being a variable.)
\begin{itemize}
\item Type- A: These are sum gates in $F$ with both inputs variables in $X$.
\item Type- B: Product gates in $F$ with both inputs variables in $X$.
\item Type- C: Sum gates in $F$ where only one input is a variable in $X$.
\item Type- D: Product gates in $F$ where only one input is a variable in $X$.
\end{itemize}

Thus, type-D gates compute polynomials of the form $h\cdot x_i$ where $x_i \in X, h \in \mathbb{F}[X\setminus\{x_i\}]$ are the inputs to the type-D gate. Let $a,b,c,d$ be the number of gates of type-A, B, C and D respectively. Let $a''$ be the number of Type $A$ gates that compute a polynomial of rank $2$ under an equi-partition $\varphi$, and $a'$ be the number of Type-$A$ gates that compute a polynomial of rank $1$ under $\varphi$ such that $a=a'+a''$.

The following lemma is an adaptation, for our distribution $\distr$, of the same lemma for the distribution of all equi-partitions on $n$ variables from \cite{RR19a}.

\begin{lemma}\label{lem:rr16}
Let $f \in \mathbb{F}[X]$ be an ROP, and $\varphi$ be an equi-partition function sampled uniformly at random from the distribution $\distr$. Then with probability at least $1 - 2^{-\Omega(n)}$, $\rank(M_f)\le 2^{n/2-\Omega(n)}$.
\end{lemma}

\begin{proof}
We first argue a rank upper bound for an arbitrary $f_i$. Let $\Phi_i$ be the formula computing $f_i$ with gates of the types described as above. Let $\hat{\varphi} = \varphi_1 \circ \ldots \circ \varphi_{n/r}$ sampled from the distribution ${\cal D}_B$ uniformly at random.

We use the Lemma~3.1 from \cite{RR19a} which concludes that type-$D$ gates do not contribute to the rank of a ROF. 
 
\begin{lemma}\cite[Lemma 3.1]{RR19a} \label{lem:rr16-ub}
Let $F$ be a ROF computing a read-once polynomial $f$ and $\varphi:X \to Y \cup Z$ be an partition function on $n$ variables. Then, $\rank(f) \le 2^{a''+\frac{2a'}{3}+\frac{2b}{3}+\frac{9c}{20}}$.
\end{lemma}

Intuitively,  Lemma~\ref{lem:rr16-ub} can be applied to a ROF $F$ under a distribution $\hat{\varphi} \sim \mathcal{D}_B$ as follows. If there are a large number of type D gates (say $\alpha n$, for some $0\le \alpha <1$),  then for any such equi-partition $\hat{\varphi}$, $\text{rank}_{\hat{\varphi}}(f) \le 2^{(1-\alpha)n/2}$.  A type $C$ gate, too, contributes a small value (at most $2$) to the rank compared to gates of types A and B. Thus, without loss of generality, we assume that the number of type C and D gates is at most $\alpha n$. Now our analysis proceeds as in \cite{RR19a}, only differing in the estimation of $a'', ~a'$ under an equi-partition $\hat{\varphi} \sim \mathcal{D}_B$.

Let $(P_1, \ldots, P_t)$ be a pairing induced by the gates of types A and B (\ie the two inputs to a gate of type A or B form a pair). There can be at most $n/2$ pairs, but since we have $\alpha n$ gates of type C and D for some $0\le \alpha <1$, we assume $(1-\alpha)n$ remaining type A and B gates. Therefore, for $t=(n-\alpha n)/2, ~t\le n/2$, we have the pairs $P_1, \ldots, P_t$ induced by the type A and B gates in $\Phi_i$. 

Now, considering the division of $X$ into $B_1,\ldots,B_{n/r}$, we can divide the pairs into two sets depending on whether a pair lies entirely within a block $B_i, ~i\in [n/r]$ or the pair has its members in two different blocks $B_i$ and $B_j$ for $i,j \in [n/r], ~i \neq j$. We define these two sets as $W=\{P_i~\mid~P_i=(x,y), \exists \ell, ~ x,y\in B_\ell \}$ for pairs lying within blocks and $A=\{P_i~\mid~P_i=(x,y), \exists j,k, ~j\neq k, ~ x\in B_j,y\in B_k\}$ for pairs lying across blocks, where $x,y$ are two arbitrary variables in $X$.

Each pair $P_i$ can be monochromatic or bichromatic under the randomly sampled equi-partition $\hat{\varphi}$ with the probability $\frac{1}{2}$. Presence of monochromatic edges will give us a reduction in the rank of $f_i$ under $\hat{\varphi}$. The analysis on $W$ and $A$ is done separately as follows.

\paragraph*{Analysing $W$, $|W|> t/2$:} Let $B_{i_1}, \ldots, B_{i_\ell}$ be the blocks containing at least one pair from $W$, $\ell \le n/r$. We want to estimate $\ell$ and count how many of these $\ell$ blocks have at least one monochromatic pair under $\hat{\varphi}$ from $W$.

For each $B_i, ~i\in [t]$, we define the Bernoulli random variable $X_i$ such that,
\[X_i = \begin{cases}1, ~\mbox{if }\exists P\in W, ~P=(x,y), ~x,y \in B_i,\\
0, ~\mbox{otherwise}. \end{cases}\] 
Let ${\sf Pr}[X_i=1]= {\sf Pr}[\exists P\in W, ~P=(x,y), ~x,y \in B_i] = \epsilon$, for some $\epsilon>0$. 

Then we have ${\sf E}[X_i]=\epsilon$, and for $\mathcal{X}=X_1+\ldots +X_{n/r}$, ${\sf E}[\mathcal{X}] =\epsilon \cdot n/r$. By the Chernoff's bound defined in \cite{MU}, we have,
\[{\sf Pr}[\mathcal{X} > 2\epsilon n/r] < \text{exp}(\frac{-\epsilon n}{3r}).\]
Now we estimate $\epsilon$ as follows:
\begin{align*}
\epsilon &= {\sf Pr}[X_i=1]= {\sf Pr}[\exists P\in W, ~P=(x,y), ~x,y \in B_i] \\
&= {\sf Pr}[x,y \in B_i | \exists P\in W, ~P=(x,y)]\\
&= \frac{{\sf Pr}[x,y \in B_i]}{{\sf Pr}[\exists P\in W, ~P=(x,y)]} \\
&\ge {\sf Pr}[x,y \in B_i] ~\mbox{ since } {\sf Pr}[\exists P\in W, ~P=(x,y)] \le 1 \\
&= \frac{1}{r^2}.
\end{align*}

Therefore, ${\sf Pr}[\mathcal{X} > 2\epsilon n/r] < \text{exp}(\frac{-\epsilon n}{3r}) \le \text{exp}(-\Omega(n)),$ when $r$ is a constant. This implies that at least $2/r^2$ fraction of the blocks have a pair entirely within them with probability $1 - \text{exp}(-\Omega(n))$ and each of these pairs is monochromatic under $\hat{\varphi}$ with the constant probability $1/2$. This gives an upper bound on the rank of $f_i$, \[\text{rank}_{\hat{\varphi}}(f_i) \le 2^{n/2 - n/r^3} = 2^{n/2 - \Omega(n)}.\]

\paragraph*{Analysing $A$, $|A|> t/2$:} Since each pair of variables in $A$ lies across two blocks, we create a graph $G=(V,E)$ where each $v_i \in V$ represents the block $B_i$ and $E=\{(v_i,v_j) ~\mid~ (x,~y)\in A, ~x\in B_i, ~y\in B_j, i\neq j\}$.

The graph $G$ has maximum degree $r$ since there can be at most $r$ pairs with one member in a fixed block $B_i$. If the edges in $E$ form a perfect matching $M'$ in $G$, then under $\hat{\varphi}$, the edges in $E$ can be either bichromatic or monochromatic. We need to show there will be sufficient number of monochromatic edges to give a tight upper bound for $\text{rank}_{\hat{\varphi}}(f_i)$.

By a result in \cite{BDDFK04}, any graph with maximum degree $r$ has a maximal matching of size $m/(2r -1)$, where $|E|=m$. Since $|A|\ge t/2$, $m \ge t/2$ and hence the maximal matching is of size $t/2(2r-1) = \Omega(n)$ when $r$ is a suitable constant. With probability $1/2$, an edge in the maximal matching 
is bichromatic. Hence, $\le t/2$ number of the edges in the maximal matching are bichromatic with probability $1/2^{t/2} = O(\text{exp}(n^{-1}))$. So, with the high probability of $1 - O(\text{exp}(n^{-1}))$, more than half of the edges in the maximal matching are monochromatic, thus giving us the rank bound,
\[\text{rank}_{\hat{\varphi}}(f_i) \ge 2^{n/2 - t/2} = 2^{n/2 - \Omega(n)}.\]

\end{proof}
Given an upper bound on the rank of ROFs under a random partition from $\distr$, we now proceed to prove the Theorem~\ref{thm:sumrof} by showing a lower bound on the size of ROFs computing our hard polynomial $h$.
\begin{proof}(Proof of Theorem~\ref{thm:sumrof})
By Observation~\ref{obs:full-rank}, the upper bound on the rank of ROFs given by Lemma~\ref{lem:rr16} and the sub-additivity of rank, we have:
\[s\cdot 2^{n/2 - \Omega(n)} \le 2^{n/2} \implies s = 2^{\Omega(n)}.\]
\end{proof}

With this result, the relationship between the classes of polynomials computable by polynomial size ROFs, ROABPs and depth-$3$ multilinear circuits is clear. Since the class of smABPs of polynomial size is strictly smaller than the class of polynomial size multilinear circuits (as in the non-multilinear setting), in the next section we obtain a lower bound on the sum of ROABPs computing the explicit polynomial in \cite{DMPY12}, which is efficiently computable by smABPs.

\section{A separation between Sum of ROABPs and smABPs}
\label{sec:sumroabp-lb}
In this section we prove a sub-exponential lower bound against the size of sum of read-once oblivious ABPs computing the hard polynomial constructed in \cite{DMPY12}. This shows a sub-exponential separation between syntactically multilinear ABPs and sum of ROABPs. 

We prove the following theorem in this section:
\sumroabp*
Our aim is to give an upper bound on the maximum rank of ROABPs under an arc partition.
We refer to the rank of the coefficient matrix of the sum of ROABPs against an arc-partition as the {\em arc-rank}. We analyze the arc-rank of the sum of ROABPs against an arc-partition to give a lower bound on the size of the sum necessary to compute $\widehat{f}$.

Let us assume that $n$ is even. In order to prove the lower bound, we need to estimate an upper bound on the arc-rank computed by a ROABP. We define the notion of $F$-arc-partition, $F$ being a ROABP, as follows:
\begin{definition}
Let us consider an arc partition $Q$ constructed from a ROABP $F$ in the following manner:
Let the order of variables appearing in the ROABP be $x_{\sigma(1)}, x_{\sigma(2)}, \ldots, x_{\sigma(n)}$, where $\sigma \in S_n$ is a permutation on $n$ indices. Then, $Q= \{(x_{\sigma(i)}, x_{\sigma(i+1)})~\mid~ i \in [n], ~i \text{ is odd}\}$ is a $F$-arc-partition.
\end{definition}


We assume $2K\mid n$. Let $S_1,\ldots,S_K$ be a $K$-coloring of the variable set $X$, where $x_1,\ldots,x_n$ are ordered according to the ROABP and for every $i\in [k]$, $S_i$ contains the variables $x_{(i-1)n/K +1},\ldots, x_{in/K}$ according to that ordering. Then $S_1, \ldots,S_K$ is a $K$-partitioning of the pairs in the $F$-arc-partition $Q$. So pairs in $Q$ are monochromatic, whereas the pairs $(P_1,\ldots,P_{n/2})$ on which a random arc-partition $\Pi$ sampled from ${\cal D}$ is based, might cross between two colors.

Our analysis for the ROABP arc-rank upper bound follows along the lines of the analysis for the arc-rank upper bound given by \cite{DMPY12} for syntactic multilinear formulas. For this analysis we define the set of violating pairs for each color $c$, $V_c(\Pi)$, that is defined as:
$$V_c(\Pi)=\{\Pi_t ~\mid~ |\Pi_t \cup S_c|=1, ~t\in [n/2] \},$$ 
where $\Pi_1, \ldots, \Pi_{n/2}$ are pairs in $\Pi$. The quantity $G(\Pi)=|\{c ~\mid~ |V_c(\Pi)| \ge n^{\frac{1}{1000}}\}|$, representing the number of colors with many violations, is similarly defined. We use the following lemma directly from \cite{DMPY12}:

\begin{lemma} \label{lem:dmpy}
Let $K \le n^{\frac{1}{100}}$, $\Pi$ be the sampled arc-partition, and $G(\Pi)$ be as defined above. Then, we have,
${\sf Pr}_{\Pi \in \mathcal{D}}[G(\Pi) \le K/1000] \le n^{-\Omega(K)}.$
\end{lemma}

The following measure is used to compute the arc-rank upper bound for ROABPs.
\begin{definition}(Similarity function)
Let $\varphi$ be a distribution on functions $\mathcal{S} \times \mathcal{S} \to \mathbb{N}$, such that $S$ is the support of the distribution on arc-partitions, $\mathcal{D}$. Let $P,Q$ be arc-partitions sampled independently and uniformly at random from $\mathcal{D}$. Then, $\varphi(Q,P):\mathcal{S} \times \mathcal{S} \to \mathbb{N}$ is the total number of common pairs between two arc-partitions $Q$ and $P$. 
\end{definition}
We assume $Q$ to be the $F$-arc-partition for the ROABP $F$.  For a pair that is not common between $\Pi$ and $Q$, we show both the variables in the pair is in the same partition, $Y$ or $Z$ with high probability. 
\begin{restatable}{theorem}{roabpup}
 \label{thm:roabp-ub}
Under an arc-partition $\Pi$ sampled from $\mathcal{D}$ uniformly at random, if $p \in \mathbb{F}[X]$ is the polynomial computed by a ROABP $P$, then, for the similarity function $\varphi$ and $\delta>0$,
$${\sf Pr}_{\Pi \sim \mathcal{D}}[\varphi(\Pi, Q) \ge n/2 - n^\delta]\le 2^{-o(n)}.$$
\end{restatable}
{\em Proof Outline:} Our argument is the same as \cite{DMPY12}. It is being included here for completeness for the parameters here being somewhat different than \cite{DMPY12}.

In order to analyse the number of common pairs counted by $\varphi$, we consider the $K$-coloring of $F$ and show that under a random arc-partition $\Pi$, the number of crossing pairs are large in number using Lemma \ref{lem:dmpy}. Then, we show, this results in large number of pairs having both elements in $Y$. 
In order to identify the colors with the high number of crossing pairs, a graphical representation of the color sets is used.

\begin{proof} 
\cite{DMPY12} construct the graph $H(\Pi)$, where each vertex is a color $c$ such that $|V_c(\Pi)| \ge n^{\frac{1}{1000}}$, and vertices $c$ and $d$ have an edge connecting them if and only if $|V_c(\Pi) \cap V_d(\Pi)|\ge n^{\frac{1}{1500}}$. We know for any two colors $c,d \in [K]$, $|V_c(\Pi) \cap V_d(\Pi)|\le n^{\frac{1}{1000}}$. So, by definition of $H(\Pi)$, the least degree of a vertex in $H(\Pi)$ is $1$. Using this, \cite{DMPY12} prove the following claim:
\begin{claim}\label{claim:dmpy12-graph}
Let the size of the vertex set of $H(\Pi)$, $V(H(\Pi))$, be $M$. For any subset $U$ of $V(H(\Pi))$ size $N \ge M/2 -1$, there is some color $h_{j+1}, j \in [N-1]$ such that in the graph induced on all vertices except $\{h_1, \ldots,h_j\}$, the degree of $h_{j+1}$ is at least $1$.
\end{claim}

By Claim~\ref{claim:dmpy12-graph}, we have $U \subseteq V(H(\Pi))$, $U=\{c_1,\ldots, c_{M/2-1}\}$ such that this is the set of colours having high number of crossing pairs common with colors not in $U$. Considering the colors sequentially, given $\Pi$, we first examine the pairs crossing from color $c_1$ to other colors, then $c_2$ and so on. Therefore, to examine the event $E_i$ for color $c_i$, we have to estimate ${\sf Pr}_{\Pi \sim \mathcal{D}}[E_i ~\mid~ E_1, \ldots, E_{i-1}, \Pi]$.

Here, $E_i$ is the event $|Y_{c_i} - |S_{c_i}|/2| \le n^{\frac{1}{5000}}$, equivalently expressed as $|S_{c_i}|/2 - n^{\frac{1}{5000}} \le Y_{c_i} \le |S_{c_i}|/2 - n^{\frac{1}{5000}}$. But for an upper bound, it suffices to analyse the $n^{\frac{1}{1500}}$ crossing pairs from $S_{c_i}$ to $S_{c_j}$ instead of considering the entire set. Let the subset of $Y_{c_i}$ constituted by one end of crossing pairs going to color $c_j$ be $P_{ij}$. Each element $x$ in a crossing pair $P_t=(x,w)$ is a binomial random variable in a universe of size $\ge n^{\frac{1}{1500}=s}$ with probability $1/2$ of being allotted to the subset $Y$ of the universe. This event is independent of how the $c_i$ colored element of other crossing pairs $P_{t'}$ are allotted. So, $|B_{ij}|=b_j$ is a hypergeometric random variable where $B_{ij}$ contains all such $x \in Y$. By the properties of a hypergeometric distribution, ${\sf Pr}_{b_j}[b_j=a] = O(s^{\frac{-1}{2}})=O(n^{\frac{-1}{3000}})$, where $a$ is a specific value taken by the size of $B_{ij}$.

Applying the union bound over all colors $c_j$ for the crossing pairs, and taking $b=\sum_{j\in U \setminus \{i\}} b_j$, we have:
$${\sf Pr}_b[s/2 - n^{\frac{1}{5000}} \le b \le |S_{c_i}|/2 - n^{\frac{1}{5000}}] \le2n^{\frac{1}{5000}}O(n^{\frac{-1}{3000}}) = n^{-\Omega(1)}.$$
Therefore, ${\sf Pr}_{\Pi \sim \mathcal{D}}[E_i ~\mid~ E_1, \ldots, E_{i-1}, \Pi] = n^{-\Omega(\delta)}$.

We want an upper bound for ${\sf Pr}[|Y_c - |S_c|/2| \le n^{\frac{1}{5000}} \forall c \in [K]]$. We have calculated an upper bound for the colors in $[K]$ that were highly connected to each other in $H(\Pi)$. So, we can now estimate the total probability as follows:
\begin{align*}
&{\sf Pr}[|Y_c - |S_c|/2| \le n^{\frac{1}{5000}} \forall c \in [K]] \\
&= {\sf E}[n^{-\Omega(G(P))}~\mid~ G(P)>K/1000] + {\sf E}[n^{-\Omega(G(P))}~\mid~ G(P)\le K/1000] \\
&= {\sf E}[n^{-\Omega(G(P))}~\mid~ G(P)>K/1000] + n^{-\Omega(K)} \text{ by Lemma \ref{lem:dmpy}}\\
&\le n^{-\Omega(K)}.
\end{align*}
If we consider $\delta=1/5000$, then:
\[ {\sf Pr}_{\Pi \sim \mathcal{D}}[\varphi(\Pi, Q) \ge n/2 - n^\delta] \le {\sf Pr}[|Y_c - |S_c|/2| \le n^{\frac{1}{5000}} \forall c \in [K]] \le n^{-\Omega(K)}\]
Now, in Lemma ~\ref{lem:dmpy}, $K\le n^{\frac{1}{1000}}$. 

Hence, ${\sf Pr}_{\Pi \sim \mathcal{D}}[\rank(M(p_\Pi))\ge 2^{n/2-n^\delta}] \le 2^{-cn^{\frac{1}{1000}}\log n}= 2^{-o(n)}$.
\end{proof}

Now, using the above Theorem ~\ref{thm:roabp-ub}, we can prove the lower bound on the size of the sum of ROABP, $s$.

\begin{proof}(of Theorem \ref{thm:sumroabp-lb})
Since the polynomial $f$ is such that each multiplicand is of the form $\lambda_e(x_u+x_v)$, if $x_u,x_v$ are both mapped to the same partition $Y$ or $Z$, it will reduce the rank of the partial derivative matrix by half. Hence, we have the following:
\[{\sf Pr}_{\Pi \sim \mathcal{D}}[\rank(M(f_\Pi))\ge 2^{n/2-n^\delta}] = {\sf Pr}_{\Pi \sim \mathcal{D}}[\varphi(\Pi, Q) \ge n/2 - n^\delta],\]
for some suitable $\delta>0$.
\begin{align*}
{\sf Pr}[{\sf rank}(M(f_\Pi))=2^{n/2}] &\le {\sf Pr}[\exists i\in [s], ~{\sf rank}(M((f_i)_\Pi))\ge 2^{n/2}/s] \\
&\le \sum_{i=1}^s {\sf Pr}[{\sf rank}(M((f_i)_\Pi))\ge 2^{n/2}/s]\\
&\le \sum_{i=1}^s {\sf Pr}[{\sf rank}(M((f_i)_\Pi))\ge 2^{n/2 - n^\delta}] \text{ for some }\delta>0 \\
&\le s \cdot n^{-\Omega(n^{\frac{1}{1000}})}\\
\implies s &= 2^{\Omega(n^{\frac{1}{1000}}\log n)}=2^{\Omega(n^{\frac{1}{500}})}.
\end{align*}
\end{proof}

The difference in computational power of ROABPs and smABPs highlights the power of reads of variables. From their definition, strict-interval ABPs generalise ROABPs by reading an interval of variables in every sub-program instead of reading a subset of variables in a fixed order. However, in the following section, we note that reading in intervals do not lend more computational power, and that ROABPs and Strict-interval ABPs in fact compute the same class of polynomials.  

\section{Strict-Interval ABPs}
\label{sec:sintABP}

A strict-interval ABP, defined in \cite{RR19} (See Definition ~\ref{def:sint-ABP}), is a restriction of the notion of interval ABPs introduced  by \cite{AR16}. In the original definition given by~\cite{RR19},  every sub-program in a strict-interval ABP $P$ is defined on a $\pi$-interval of variables for some order $\pi$, however, without loss of generality,  we assume $\pi$ to be the identity permutation on $n$ variables. Therefore, an interval of variables $[i,j], ~i<j$ here is the set $\{x_i, \ldots,x_j\}$.  In this section we show that strict-interval ABPs are equivalent to  ROABPs upto a polynomial blow-up in size. 

\begin{theorem}
\label{thm:sint-roabp}
The class of strict-interval ABPs is equivalent to the class of ROABPs.
\end{theorem}

The proof of Theorem~\ref{thm:sint-roabp} involves a crucial observation that in a strict-interval ABP, variables are read in at most two orders and the nodes that correspond to paths that read in different orders can be isolated. We start with some observations on intervals in $[1,n]$ and the intervals involved in a strict interval ABP. 

Let $P$ be a strict-interval ABP over the variables $X=\{x_1,\ldots, x_n\}$. For any two nodes $u$ and $v$ in $P$, let $I_{u,v}$ be the interval of variables associated with the sub-program of $P$ with $u$ as the start node and $v$ as the terminal node. For two intervals $I =[a,b], J =[c,d]$ in $[1,n]$, we say  $I \preceq J$, if $b\le c$. Note that any two intervals $I$ and $J$ in $[1,n]$ are comparable under $\preceq$ if and only if  either they are disjoint or the largest element in one of the intervals is the smallest element in the other. This defines a natural transitive relation on the set of all intervals in $[1,n]$. 
The following is a useful property of  $\preceq$:
\begin{observation}
\label{obs:int}
Let $I, J$ and $J'$ be intervals over $[1,n]$ such that $I \preceq J$ and $J' \subseteq J$. Then $I \preceq J'$. 
\end{observation}
\begin{proof}
Let $I = [a,b], J =[c,d]$ and $J' = [c',d']$. As $I \preceq J$, we have $b \le c$. Further, since $J' \subseteq J$, we have $c\le c'$ and $d'\le d$. Therefore, $b\le c'$ and hence $I \preceq J'$.
\end{proof}
We begin with an observation on the structure of intervals of the sub-programs of $P$.
Let $v$ be a node in $P$. We say $v$ is an {\em ascending} node, if $I_{s,v}\preceq I_{v,t}$ and a {\em descending} node if $I_{v,t}\preceq I_{s,v}$.
\begin{observation}
Let $P$ be a strict-interval ABP and $v$ any node in $P$. Then, $v$ is either ascending or descending and not both.
\end{observation}
\label{obs:ascdesc}
\begin{proof}
Let $I= I_{s,v}$ and $J = I_{v,t}$. Since $P$ is a strict-interval ABP, the intervals $I$ and $J$ are disjoint and hence either $I \preceq J$ or $J \preceq I$ as required.
\end{proof}
Consider any $s$ to $t$ path $\rho$ in $P$. We say that $\rho$ is {\em ascending} if every node in $\rho$ except $s$ and $t$ is ascending. Similarly, $\rho$ is called descending if every node in $\rho$ except $s$ and $t$ is descending.

\begin{lemma}
\label{lem:path-order}
Let $P$ be a strict interval ABP and let $\rho$ any $s$ to $t$ path in $P$. Then either $\rho$ is ascending or descending.
\end{lemma}
\begin{proof}
We prove that no $s$ to $t$ path in $P$ can have both ascending and descending nodes. For the sake of contradiction, suppose that $\rho$ has both ascending and descending nodes. There are two  cases.  In the first, there is an edge  $(u,v)$ in $\rho$ such that $u$ is an ascending node and $v$ is a descending node.  Let $I = I_{s,u}, J = I_{u,t}, I' = I_{s,v}$ and $J'=I_{v,t}$.
Since $P_{s,u}$ is a sub-program of $P_{s,v}$, we have $I \subseteq I'$, similarly $J' \subseteq J$.  By the assumption,  we have $I\preceq J$ and $J' \preceq I'$. By Observation~\ref{obs:int}, we have $I \preceq J'$ and $J' \preceq I'$. By transitivity, we have $I \preceq I'$. However, by the definition of $\preceq$, $I$ and $I'$ are incomparable, which is a contradiction. The second  possibility is $u$ being a descending node and $v$ being an ascending node. In this case, $J \preceq I$ and $I' \preceq J'$. Then, by Observation~\ref{obs:int}, we have $J' \preceq I$ as $J' \subseteq J$. Therefore, $J \preceq J'$ by the transitivity of $\preceq$, a contradiction.  This completes the proof.
\end{proof}
Lemma~\ref{lem:path-order} implies that the set of all non-terminal nodes of $P$ can be partitioned into two sets such that there is no edge from one set to the other. Formally:

\begin{lemma}
\label{lem:interval-two-order} Let $P$ be an interval ABP. There exist two  strict-interval ABPs $P_1$ and $P_2$ such that 
\begin{enumerate}
\item All non-terminal nodes of $P_1$ are ascending nodes and all non-terminal nodes of $P_2$ are descending nodes; and
\item $P = P_1 + P_2$. 
\end{enumerate}
\end{lemma}

\begin{proof}
Let $P_1$ be the sub-program of $P$ obtained by removing all descending nodes from $P$ and $P_2$ be the sub-program of $P$ obtained by removing all ascending nodes in $P$. By Lemma~\ref{lem:path-order}, the non-terminal nodes in $P_1$ and $P_2$ are disjoint and every $s$ to $t$  path $\rho$ in $P$ is either a $s$ to $t$ path in $P_1$ or a $s$ to $t$ path in $P_2$ but not both. Thus $P = P_1+ P_2$.
\end{proof}

Next we show that any strict-interval ABP consisting only of ascending or only of descending nodes can in fact be converted into an ROABP. 
\begin{lemma}
\label{lem:asc-roabp}
Let $P$ be a strict-interval ABP consisting only of ascending nodes or only of descending nodes. Then the polynomial computed by $P$ can also be computed by a ROABP $P'$ of size polynomial in ${\sf size}(P)$. The order of variables in $P'$ is $x_1,\dots, x_n$ if $P$ has only ascending nodes and $x_n, \ldots, x_1$ if $P$ has only descending nodes.   
\end{lemma}
\begin{proof}
We consider the case when all non-terminal nodes of $P$ are ascending nodes.   Let $\rho$ be any $s$ to $t$ path in $P$.  We claim that the edge labels in $\rho$ are according to the order $x_1,\ldots, x_n.$ Suppose that there are edges $(u,v)$ and $(u',v')$ occurring in that order in $\rho$ such that $(u,v)$ is labelled by $x_i$ and $(u',v')$ is labelled by $x_j$ with $j<i$. Let $I' = I_{s,u'} $ and $J' = I_{u',t}$.  Since $i \in I'$, $j \in J'$ and $I' \cap J' = \emptyset$, it must be the case that $J'\preceq I'$ and hence $u'$ must be a descending node, a contradiction. This establishes that $P$ is an one ordered ABP. By the equivalence between one ordered ABPs and ROABPs (\cite{Jan08}, \cite{JQS10}), we conclude that the polynomial computed by $P$ can also be computed by a ROABP of size polynomial in the size of $P$. 

The argument is similar when all non-terminal nodes of $P$ are descending. In this case, we  have $i<j$ in the above argument and hence $I' \preceq J'$, making $u'$ an ascending node leading to a contradiction. This concludes the proof.
\end{proof}
A permutation $\pi$ of $[1,n]$ naturally induces the order $x_{\pi(1)}, \ldots, x_{\pi(n)}$. The {\em reverse} of $\pi$ is the order $x_{\pi(n)}, x_{\pi(n-1)}, \ldots, x_{\pi(1)}$.  Since branching programs are layered, any multilinear polynomial computed by a ROABP where variables occur in the order given by $\pi$ can also be computed by a ROABP where variables occur in the reverse of $\pi$.
\begin{observation}
\label{obs:reverse}
Let $P$ be a ROABP where variables occur in the order induced by a permutation $\pi$. The polynomial computed by $P$ can also be computed by a ROABP of same size as $P$ that reads variables in the reverse order corresponding to $\pi$.
\end{observation}
\begin{proof}
Let $P'$ be the ROABP obtained by reversing the edges of $P$ and swapping the start and terminal nodes. Since $P$ is a layered DAG, there is a bijection between the set of all $s$ to $t$ paths in $P$ and the set of all $s$ to $t$ paths in $P'$, where the order of occurrence of nodes and hence the edge labels are reversed. This completes the proof.  
\end{proof}

The above observations immediately establish Theorem~\ref{thm:sint-roabp}.
\begin{proof}[Proof of Theorem~\ref{thm:sint-roabp}]
Let $P$ be a strict-interval ABP of size $S$ computing a multilinear polynomial $f$. By Lemma~\ref{lem:interval-two-order} there are strict interval ABPs $P_1$ and $P_2$ such that $P_1$ has only ascending non-terminal nodes and $P_2$ has only descending non-terminal nodes such that $f = f_1 + f_2$ where $f_i$ is the polynomial computed by $P_i$, $i\in \{1,2\}$. By  Lemma~\ref{lem:asc-roabp} and Observation~\ref{obs:reverse}, $f_1$ and $f_2$  can be computed by a ROABPs    that read the variables in the order $x_1,\ldots, x_n$. Then $f_1 + f_2$ can also be computed by an ROABP. 
It remains to bound the size of the resulting ROABP. Note that ${\sf size }(P_i) \le S$. A ROABP for $f_i$ can be obtained by staggering the reads of $P_i$ which blows up the size of the ABP by a factor of $n$ (\cite{Jan08}, \cite{JQS10}). Therefore size of the resulting ROABP is at most $2nS \le O(S^2)$.
\end{proof}

Using Theorem~\ref{thm:sint-roabp}, we can design the following white-box PIT for strict-interval ABPs.
\begin{corollary}
\label{cor:sintABP-wbPIT} Given a strict-interval ABP $P$ of size $s$, we can check whether the polynomial computed by $P$ is identically zero in time $O(\text{poly}(S))$. 
\end{corollary}
\begin{proof}
The proof follows from Theorem~\ref{thm:sint-roabp} and the polynomial time white-box \PIT algorithm given by \cite{RS05} for non-commutative ABPs, since the variables in $X$ are read only once, in a fixed order, in a ROABP.
\end{proof}


The notion of intervals of variables corresponding to every sub-program can be applied to formulas in the form of Interval Formulas, where every sub-formula corresponds to an interval. In the following section we explore how such a model can be used to generalize the model of ROFs, and in what ways it differs from ROFs.

\section{Interval Formulas}
\label{sec:intformula}
We saw that  strict-interval ABPs have the same computational power as ROABPs despite being seemingly a non-trivial generalization. It is naturally tempting to guess that a similar generalization of ROFs might yield a similar result. 
However, we observe that such a generalization of ROFs yields a class different from  ROFs.  

We introduce interval formulas as a generalization of read-once formulas. An interval on variable indices, $[i,j],~i<j$, is an interval corresponding to the set of variables $X_{ij} \subseteq X =\{x_1, \ldots, x_n\}$, where $X_{ij}= \{x_p \mid x_p \in X, ~i\le p\le j\}$. Polynomials are said to be defined on the interval $[i,j]$ when the input variables are from the set $X_{ij}$. When there is no ambiguity, we refer to $X_{ij}$ as an interval of variables $[i,j]$. 
Gates in a read-once formula $F$ can also be viewed as reading an interval of variables according to an order $\pi$ on the variables \ie there is a permutation $\pi \in S_n$ such that every gate $v$ in $F$ is a sub-formula computing a polynomial on a $\pi$-interval of variables. Thus, interval formulas are a different generalization of read-once formulas where every gate $v$ in the formula $F$ reads an interval of variables in a fixed order.
 
We formally define interval formulas 
as follows:  

\begin{definition}(Interval Formulas)
An arithmetic formula $F$ is an {\em interval formula} if for every gate $g$ in $F$, there is an interval $[i,j], ~i<j$ such that $g$ computes a polynomial in $X_{ij}$ and for every product gate $g=h_1 \times h_2$, the intervals corresponding to $h_1$ and $h_2$ must be non-overlapping. 
\end{definition}
Thus, if a product gate $g$ in $F$ defined on an interval $I=[i,j]$ takes inputs from gates $g_1,\ldots,g_t$, then the gates $g_1,\dots,g_t$ compute polynomials on disjoint intervals $[i,j_1], [j_1+1,j_2], \ldots, [j_{t-1}+1,j]$ respectively, where $\forall p, ~j_p<j_{p+1}$ and $i \le j_p \le j$. If $g_1,g_2$, defined on intervals $I_1,~I_2$ are input gates to a sum gate $g'$, then the interval $I$ associated with $g'$ is $I=I_1 \cup I_2$.

A quick observation is that interval formulas are different from ROFs:
\begin{proposition}
\label{prop:interval-rof}
The set of all polynomials computable by interval  formulas is different from that of ROFs
\end{proposition}
\begin{proof}
By~\cite{Vol16}, the polynomial $x_1x_2+x_2x_3 + x_1x_3$ is not an ROF. However, the expression $x_1x_2+x_2x_3 + x_1x_3$ is itself an interval formula.
\end{proof}

 In fact, interval formulas are universal, since any sum of monomials can be represented by an interval formula. 
 
 Our next observation is that  the polynomial $f_{\sf PRY}$ defined in Section~\ref{sec:sumrof} can be computed  by an interval formula.
\begin{restatable}{proposition}{pryinterval}
\label{prop:PRY-int} The polynomial family $f_{\sf PRY}$ is computable by  an interval formula of polynomial size.
\end{restatable}
\begin{proof}
Recall that $f_{\sf PRY}(X) = f_{\sf RY}(B_1) \cdot f_{\sf RY}(B_2)\cdots f_{\sf RY}(B_{n/r})$. Since each of the $f_{\sf RY}(B_i)$ is a constant variate polynomial and the sum of product representation of any multilinear polynomial is an interval formula by definition, we have that $f_{\sf RY}(B_i)$ is computable by an interval formula of constant size. This $f_{\sf PRY}(X)$ has a polynomial size interval formula.
\end{proof}
It is not known if every ROF can be converted to a ROF of logarithmic depth. However, we argue, in the following that interval formulas can be depth-reduced efficiently.

We have the following depth reduction result for general arithmetic formulas given by \cite{Bre74}, who showed that depth of any arithmetic formula can be reduced by allowing its size to be increased by a  polynomial factor.
\begin{theorem}\label{thm:brent74}
\cite{Bre74} Any polynomial $p$ computed by an arithmetic formula of size $s$ and depth $d$, can also be computed by a formula of size $\text{poly}(s)$ and depth $O(\log s)$.
\end{theorem}

We know that  this reduction preserves multilinearity. However, we don't know if Theorem~\ref{thm:brent74} can be modified to preserve the read-$k$ property. We show that the depth reduction algorithm given by Theorem~\ref{thm:brent74} preserves the interval property.

\begin{restatable}{theorem}{depthred}
\label{thm:intf-depthred}
Let $f \in \mathbb{F}[X]$ be a polynomial computed by an interval formula $F$ of size $s$ and depth $d$. Then $f$ can also be computed by an interval formula of size $\text{poly}(s)$ and depth $O(\log s)$.
\end{restatable}
\begin{proof}
We know that the underlying structure of any arithmetic formula is a tree. The proof by Brent crucially uses the fact that by the {\em tree-separator lemma} \cite{Chu89}, we are guaranteed that there exists a tree-separator node $g$ such that the sub-tree $\Phi$ of a formula $\Phi'$ of total size $s$, rooted at the node $g$, has size $\le 2s/3$.

The construction by \cite{Bre74} proceeds as follows. We replace the gate $g$ by a new formal variable $y$. Let the resulting polynomial computed by $F$ be $f'(x_1,\ldots,x_n,y)$, where $f(x_1,\ldots,x_n)=f'(x_1,\ldots,x_n,g)$ under the new substitution $y=g$. As $f'$ is linear in $y$, we have 
$$ f'(x_1,\ldots,x_n,y) = yf_1(x_1,\ldots,x_n)+ f_0(x_1,\ldots,x_n),$$ where $f_0 = f'\mid_{y=0}$ and $f_1= f'\mid_{y=1} - f'\mid_{y=0}$. Thus, $f_0, f_1$ can be computed by multilinear formulas of size less than $\text{size}(F)$. Now, recursively obtaining small-depth formulas for $f_1,f_0$, we obtain a $O(\log s)$ depth formula computing $f$.

However, the above construction does not necessarily preserve the interval property, since the intervals of variables on which $f_0,f_1$ and $g$ are defined, can be overlapping. We overcome this problem by expressing $f_0,f_1$ as products of polynomials over disjoint intervals, each of the intervals being disjoint to the interval corresponding to $g$.

We assume, without loss of generality, that the interval formula $F$ corresponds to the interval $[1,n]$. Let the interval corresponding to $g$ be $I_g=[i,j], ~i<j$. Now, by definition of $f_1$ and $f_0$, they are defined on the same interval of variables. We consider the intervals $I_0, I_1$ such that $I_0 \cup I_1 = [1,n]\setminus [i,j]$, $I_0=[j+1,n]$ and $I_1=[1,i-1]$. We express both $f_0,f_1$ as products of two polynomials on $I_0$ and $I_1$ respectively. As $f_1$ and $g$ are multiplicatively related in $F$, we show that $f_1= f_{1,1}\times f_{1,0}$ where $f_{1,1}$ is a polynomial on the interval $I_1$ and $f_{1,0}$ is a polynomial on the interval $I_0$.

We consider the root to leaf ($g$) path $\rho$ in the original formula $F$ containing the node $g$. All the paths meeting $\rho$ at a sum gate represent polynomials additively related to $y$ \ie contributing towards the computation of $f_0$ and not $f_1$. For $f_1$, we will analyze only the paths meeting $\rho$ at product gates. Let us consider a product gate on $\rho$ computing $h_1 \times h_2$, such that $h_2$ lies on $\rho$. Since $I$ is contained in the interval corresponding to $h_2$, the interval corresponding to $h_1$, $I_{h_1}$ must be either fully contained in $I_1$ or $I_0$. 
\subparagraph*{Constructing an interval formula for $f_1$:}
We ignore all sum gates on $\rho$ computing $p_1+p_2$, with $p_2$ on $\rho$, by substituting $p_1$ to zero. The resulting formula is $F'$.
In any product gate computing $h_1 \times h_2$, where $h_2$ is on $\rho$, if $I_{h_1} \subset I_0$, we substitute $h_1$ by $1$. We also substitute $g$ by $1$. The remaining formula $F'_1$ computes the polynomial $f^{(1)}$. 

We repeat this process above, but this time, we substitute $h_1$ by $1$ only when $I_{h_1}\subset I_1$. This remaining formula $F_2'$ computes $f^{(2)}$. By definition of $f_1$, $f_1=f^{(1)}\cdot f^{(2)}$. The interval corresponding to $F_1'$ is contained in $I_1$, the interval corresponding to $F_2'$ is contained in $I_0$.   
\subparagraph*{Constructing an interval formula for $f_0$:} We ignore all product gates on $\rho$ computing $h_1 \times h_2$, with $h_2$ on $\rho$, by substituting $h_1$ by $1$. The resulting formula is $\hat{F}$.

In any sum gate computing $p_1 + p_2$, where $p_2$ is on $\rho$, if $I_{p_1} \subset I_0$, we substitute $p_1$ by $0$. We also substitute $g$ by $0$. The remaining formula $\hat{F}_1$ computes the polynomial $p^{(1)}$. 

We repeat this process from the beginning, but substitute $p_1$ by $0$ only when $I_{p_1}\subset I_1$. This remaining formula $\hat{F}_2$ computes $p^{(2)}$. By definition of $f_0$, $f_0=p^{(1)}+ p^{(2)}$. The interval corresponding to $\hat{F}_1$ is contained in $I_1$, the interval corresponding to $\hat{F}_2$ is contained in $I_0$.  

Hence, we obtain $f=f^{(1)}f^{(2)}g+p^{(1)}+p^{(2)}$. The recursive relation for calculating depth is as follows:  $\text{depth}(F)= \text{depth}(g)+2 \implies \text{depth}(s)= \text{depth}(2s/3) + 2$, which yields a total depth of $O(\log s)$ for $F$.
\end{proof}

\bibliographystyle{plainurl}
\bibliography{allref}
\end{document}